\documentclass[journal]{IEEEtran}
\usepackage{amsmath,amssymb,amsthm}
\usepackage{graphicx}
\usepackage{booktabs}
\usepackage{multirow}
\usepackage{algorithm}
\usepackage{algpseudocode}
\usepackage{cite}
\usepackage{url}
\usepackage{balance}
\usepackage[hidelinks]{hyperref}
\usepackage{mathtools}
\usepackage{array}
\usepackage{tabularx}
\usepackage{makecell}
\newtheorem{theorem}{Theorem}
\newtheorem{lemma}{Lemma}
\newtheorem{proposition}{Proposition}
\newtheorem{corollary}{Corollary}
\newtheorem{remark}{Remark}

\DeclareMathOperator{\PPP}{PPP}
\DeclareMathOperator{\argmin}{arg\,min}
\newcommand{\R}{\mathbb{R}}
\newcommand{\E}{\mathbb{E}}
\newcommand{\Prb}{\mathbb{P}}

\newcommand{\pos}[1]{\left[#1\right]_+}

\title{Spatial Prefix Caching for Wireless Edge LLM Inference:\\A Stochastic-Geometry and Queueing Framework}

\author{Le Yang, Zhouyong Liu%
\thanks{The authors are with Southeast University.}%
}

\begin{document}
\maketitle

\begin{abstract}
Prefix caching reuses the key--value (KV) states of shared prompt prefixes and can substantially reduce the time to first token (TTFT) of large language model (LLM) inference. In a wireless edge network, however, prefix states are distributed across geographically separated GPU nodes. A nearby node offers a short radio path but may provide little reuse, whereas a more distant node may cache a longer matching prefix but incur additional communication and queueing delay. Moreover, persistent prefixes and active-request KV states compete for the same GPU memory, so aggressive caching can reduce inference concurrency and create queueing hotspots. This paper develops a stochastic-geometry and queueing framework for this spatial communication--caching--computation tradeoff. We represent the prompt workload by a prefix forest and define an ancestor-closed cache profile that may contain multiple reusable prefixes. Edge GPU nodes form a Poisson point process and are independently marked by cache profile, yielding analytically tractable spatial tiers. We derive the profile-association probability, conditional serving-distance distribution, token-level computation-offloading ratio, and TTFT coverage probability under a load-aware association policy. A fixed-point formulation captures the coupling between spatial association and multi-server GPU queues, while an outer optimization selects the cache-profile distribution subject to static-memory and stability constraints. Analytical and Monte Carlo results agree closely. The results show that the latency-optimal node need not be the nearest node, that TTFT can be non-monotonic in cached-prefix depth because of GPU-memory coupling, and that queue-unaware prefix routing can collapse under skewed cache affinity. The framework converts classical binary content caching into a spatial theory of partial computation-state reuse and provides design guidelines for edge-node density, GPU-memory partitioning, and cache-aware request routing.
\end{abstract}


\section{Introduction}
\subsection{Motivation and Research Gap}
\IEEEPARstart{L}{arge} language models (LLMs) are increasingly used in interactive assistants, retrieval-augmented generation, industrial agents, and privacy-sensitive applications. Moving inference from a remote cloud toward wireless edge nodes can shorten network paths and reduce dependence on centralized infrastructure. Yet an edge GPU has limited memory and computation capacity, while long prompts and high concurrency make inference expensive.

Autoregressive inference contains a prompt-prefill phase followed by token-by-token decoding. During prefill, the model processes the input sequence and constructs the corresponding key--value (KV) states. Many production prompts share system instructions, templates, retrieved documents, tool descriptions, or multi-turn history. Recomputing these shared portions wastes GPU cycles. PagedAttention and vLLM improve KV-memory utilization and sharing~\cite{kwon2023pagedattention}; Prompt Cache explicitly reuses modular prompt states~\cite{gim2024promptcache}; RadixAttention in SGLang automatically organizes shared prefixes~\cite{zheng2024sglang}; and CachedAttention preserves conversation states across memory and storage tiers~\cite{gao2024cachedattention}. These systems establish prefix reuse as a practical serving primitive.

Beyond raw model execution, edge deployment changes the optimization objective. A cloud cluster is normally connected by a high-capacity fabric and is operated as a single resource pool. A wireless edge network is geographically fragmented: the node with the best radio path need not hold the most useful prefix, and the node with the most useful prefix may be heavily loaded. Consequently, a scheduler must decide not only \emph{whether} a prefix is reusable, but also \emph{where} that reusable computation state should be placed and \emph{when} it is worth traversing a longer network path to reach it.

The issue becomes more pronounced for agentic and retrieval-augmented workloads. A large population of requests may share system instructions, tool specifications, retrieval templates, or a popular document prefix, while their final user-specific suffixes remain unique. Such workloads naturally create a prefix forest rather than a flat catalogue of independent objects. The value of a cache entry is therefore determined jointly by its depth, the number of descendant requests that can reuse it, its persistence lifetime, and the runtime memory that it displaces.

A network-level theory is useful for three practical planning tasks. First, an operator must decide the edge-GPU density needed to satisfy a TTFT service-level objective. Second, a serving platform must divide GPU memory between persistent reusable states and active-request states. Third, a router must balance cache affinity against spatial communication cost and queue congestion. These decisions are tightly coupled and cannot be inferred from server-level cache-hit ratio alone.

The next step is distributed prefix caching. Preble jointly considers cache affinity and cluster load~\cite{srivatsa2025preble}; Mooncake treats KV states as a first-class disaggregated resource~\cite{qin2024mooncake}; ShadowServe shows that network transfer can bottleneck remote prefix reuse~\cite{xiang2025shadowserve}; and recent distributed prompt caching directly targets cooperative edge devices connected by wireless links~\cite{matsutani2026distributed}. These works motivate a fundamental wireless-edge question:

\emph{Which edge GPU should serve a request when candidate nodes differ simultaneously in radio distance, reusable-prefix depth, queueing delay, and memory-limited concurrency?}

This question is not captured by classical wireless caching. Traditional stochastic-geometry models treat cached objects as independent files, and a cache hit is approximately binary~\cite{chae2016caching,chen2017probabilistic}. KV prefixes instead have three distinctive properties.

First, prefixes form a tree or forest: a long reusable prefix contains shorter ancestors. Second, a partial hit still saves computation, so the benefit is token-level rather than binary. Third, persistent prefix states and active-request KV states occupy the same GPU memory. Caching more prefixes can therefore save prefill work while reducing runtime concurrency and increasing queueing delay.

This paper develops a spatial analytical framework for these coupled effects. Fig.~\ref{fig:concept} illustrates the central tradeoff. The typical user sees multiple candidate edge GPUs: a nearby node without reusable state, a more distant shallow-prefix node, and a still more distant deep-prefix node. The latency-optimal choice depends on the sum of radio delay, residual prefill, and queueing.

\begin{figure}[t]
\centering
\includegraphics[width=\columnwidth]{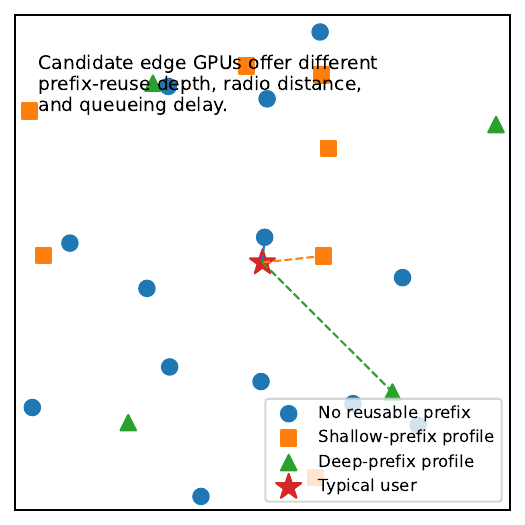}
\caption{Spatial prefix caching. Independently marked edge-GPU nodes offer different reusable-prefix depths, communication distances, and queueing delays.}
\label{fig:concept}
\end{figure}

\subsection{Contributions}
The main contributions are summarized as follows.

\begin{itemize}
\item \textbf{Prefix-forest cache profiles:} We model a reusable prompt workload as a prefix forest. Each edge node stores an ancestor-closed cache profile, allowing one spatial mark to represent multiple related prefixes and generalizing the conventional ``one file per mark'' abstraction.

\item \textbf{Load-aware spatial association:} We propose a communication--caching--computation-aware policy that minimizes predicted TTFT. Independent marking decomposes the edge network into profile-specific Poisson tiers. We derive integral expressions for profile-association probabilities and conditional serving-distance distributions.

\item \textbf{Queue-distribution-aware TTFT coverage:} We couple the spatial model to an analytically calibrated $M/M/c$ GPU queue. Association uses predicted mean waiting time, while TTFT coverage uses the full Erlang-C waiting-time distribution. A damped fixed point captures the feedback between association and queue load.

\item \textbf{GPU-memory coupling and placement optimization:} We explicitly divide GPU memory between persistent reusable prefixes and active requests. The cache-profile distribution is optimized under average static-memory and queue-stability constraints.

\item \textbf{Validated structural insights:} Analytical association probabilities match Monte Carlo simulation. Numerical results reveal a load-dependent optimal prefix depth, severe hotspot formation under queue-unaware routing, and a density-planning benefit that cannot be inferred from cache-hit probability alone.
\end{itemize}

To the best of our knowledge, prior distributed prefix-caching systems do not provide a stochastic-geometry framework that jointly analyzes spatial prefix availability, token-level prefill savings, queueing, and GPU-memory competition.

The remainder of the paper follows the same chain of reasoning. Section~II relates computation-state reuse to wireless caching and distributed LLM serving. Section~III builds the spatial, prefix, resource, and queueing model. Section~IV first characterizes association, then closes the association--load loop, and finally derives TTFT reliability and structural design implications. Section~V converts these results into trace-driven cache-profile construction and spatial placement. Section~VI validates the analysis and studies the resulting design tradeoffs. Section~VII maps the framework to a deployable control architecture and discusses extensions beyond the baseline model. The appendices provide complete proofs and reproducibility details.

\section{Background and Related Work}
The proposed framework sits at the intersection of two mature research lines: stochastic-geometry analysis of wireless caching and system-level optimization of KV-state reuse. The key gap is not a missing server-side cache primitive, but the absence of a network-level model that explains how reusable computation, spatial access, GPU memory, and queue load interact.

\subsection{Reusable Computation States and Wireless Edge Caching}
For a transformer with $N_{\rm L}$ layers, $N_{\rm H}$ KV heads, head dimension $d_{\rm H}$, numerical precision $b_{\rm e}$ bytes, and sequence length $L$, an uncompressed per-request KV footprint is approximately
\begin{equation}
M_{\rm KV}(L)=2N_{\rm L}N_{\rm H}d_{\rm H}b_{\rm e}L,
\label{eq:kvmemorybackground}
\end{equation}
where the factor two accounts for keys and values. The footprint is linear in sequence length, while the prefill computation generally grows at least linearly and includes an attention-dependent component. A cached exact prefix allows the server to skip the corresponding portion of prefill and continue from the stored state without changing model output.

Prefix reuse is different from conventional output caching. The cached object is not a completed answer but a model-internal computation state. It is useful only for requests whose token sequence follows the same prefix, and its benefit depends on the number of reusable tokens. Exact-prefix reuse is therefore naturally represented by a radix tree or prefix forest. Fig.~\ref{fig:forest} illustrates three cache profiles: broad shallow replication, a deep tool-agent branch, and a deep retrieval branch.

\begin{figure}[t]
\centering
\includegraphics[width=\columnwidth]{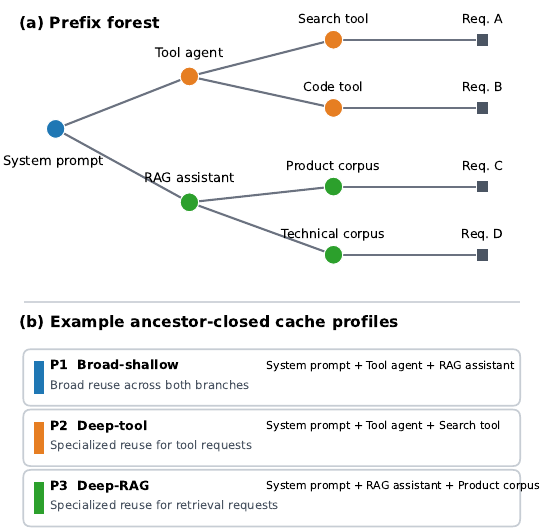}
\caption{A prefix forest and three ancestor-closed cache profiles. The upper panel shows the shared prompt structure; the lower panel lists representative broad-shallow and branch-specialized profiles. For each request, the useful hit is the deepest cached ancestor on its path.}
\label{fig:forest}
\end{figure}

The preceding mechanism turns repeated prompt processing into a cacheable resource. Wireless caching theory provides a natural starting point for asking where such a resource should be placed, but its classical object model must be revisited.

PPPs provide tractable approximations for irregular cellular deployments and permit closed-form or single-integral performance expressions through independent thinning and the probability-generating functional~\cite{andrews2011tractable,haenggi2012stochastic}. In cache-enabled networks, probabilistic placement maps nodes storing a given object to a file-specific thinned PPP. This abstraction has been used to analyze cache-hit probability, serving distance, delivery success, and throughput. Chae and Choi optimize probabilistic placement while accounting for channel-selection diversity and interference~\cite{chae2016caching}. Chen \emph{et al.} distinguish cache-hit-optimal and throughput-optimal placement in D2D networks~\cite{chen2017probabilistic}.

The classical model is powerful but relies on assumptions that are unsuitable for computation-state caching: objects are usually independent; a hit is binary; and storage does not directly reduce the server's concurrent service capacity. Our cache profile preserves the tractability of independent marking while allowing hierarchical overlap and explicit memory--queue coupling.

\subsection{KV Caching and Distributed LLM Serving}
PagedAttention reduces fragmentation and enables flexible sharing of KV blocks~\cite{kwon2023pagedattention}. Prompt Cache reuses precomputed attention states for recurring prompt modules~\cite{gim2024promptcache}. SGLang's RadixAttention organizes shared prefixes in a radix tree and exposes cache-aware scheduling opportunities~\cite{zheng2024sglang}. CachedAttention extends state persistence across memory and storage tiers for multi-turn conversations~\cite{gao2024cachedattention}. Keyformer reduces the runtime KV footprint through token selection~\cite{adnan2024keyformer}, while Marconi studies prefix admission and eviction for hybrid attention--recurrent models~\cite{pan2024marconi}. These systems demonstrate that reuse value depends on both saved computation and memory footprint, but their optimization domains are primarily a single server or a tightly connected cluster.

Server-side memory management alone does not determine edge performance. Once reusable states are distributed across multiple nodes, the scheduler must trade prefix locality against transfer cost and congestion; this has motivated a second line of work on distributed serving.

DistServe separates prefill and decode to satisfy TTFT and per-token service-level objectives~\cite{zhong2024distserve}. Splitwise independently provisions the two phases~\cite{patel2023splitwise}, Sarathi-Serve uses chunked prefill to control throughput--latency interference~\cite{agrawal2024sarathi}, and Llumnix migrates active requests to reduce tail latency under dynamic load~\cite{sun2024llumnix}. Mooncake and MemServe treat KV states as disaggregated, elastic resources~\cite{qin2024mooncake,hu2024memserve}, while KVDirect optimizes inter-node state transfer for distributed disaggregated inference~\cite{chen2024kvdirect}. Preble jointly considers prefix locality and cluster load in distributed prompt scheduling~\cite{srivatsa2025preble}. CacheGen compresses and streams KV states over bandwidth-limited paths~\cite{liu2023cachegen}. ShadowServe further shows that remote KV transfer can become network-bound and that decompression may interfere with inference~\cite{xiang2025shadowserve}. Distributed prompt caching on edge devices confirms that partial state sharing can reduce TTFT but introduces wireless catalogue and state-transfer overhead~\cite{matsutani2026distributed}. Decentralized prefix-aware routing has also been explored for peer-to-peer LLM serving, where high network latency and cache-affinity hotspots limit the gain~\cite{nair2026p2p}.

\subsection{Research Gap and Positioning}
The literature therefore provides the essential mechanisms--prefix indexing, KV sharing, disaggregation, and cache-aware scheduling--but leaves their spatial interaction largely implicit. In particular, there is no general latency law that begins with a random edge deployment, allows partial prefix matches, accounts for cache-induced traffic concentration, and exposes the opportunity cost of persistent KV memory. This is the role of the model developed in the following sections.

Table~\ref{tab:positioning} positions the proposed framework. Existing systems optimize concrete serving mechanisms and provide essential empirical motivation. The present work is complementary: it asks how random spatial deployment, partial prefix reuse, queue load, and GPU memory jointly determine network-level latency reliability.

\begin{table*}[t]
\caption{Positioning Relative to Representative Caching and LLM-Serving Work}
\label{tab:positioning}
\centering
\scriptsize
\resizebox{\textwidth}{!}{%
\begin{tabular}{p{2.65cm}ccccccp{3.1cm}}
\toprule
Work & Prefix structure & Distributed nodes & Wireless/spatial model & Queue/load coupling & GPU-memory coupling & Analytical latency law & Primary focus \\
\midrule
Classical probabilistic caching~\cite{chae2016caching,chen2017probabilistic} & No & Yes & Yes & Limited & No & Yes & File availability and delivery \\
PagedAttention / Prompt Cache~\cite{kwon2023pagedattention,gim2024promptcache} & Yes & No & No & System-specific & Yes & No & Memory management and state reuse \\
SGLang / CachedAttention~\cite{zheng2024sglang,gao2024cachedattention} & Yes & Cluster & No & Scheduler-level & Yes & No & Prefix indexing and state persistence \\
Preble~\cite{srivatsa2025preble} & Yes & Cluster & No & Yes & Indirect & No & Prefix-aware load balancing \\
Mooncake / ShadowServe~\cite{qin2024mooncake,xiang2025shadowserve} & Yes & Yes & Fabric/backhaul & System-level & Yes & No & KV disaggregation and transfer \\
Distributed edge prompt caching~\cite{matsutani2026distributed} & Partial match & Yes & Wireless testbed & Limited & Yes & No & Cooperative edge implementation \\
\textbf{This work} & Prefix forest & Yes & \textbf{PPP edge network} & \textbf{Fixed-point queues} & \textbf{Explicit} & \textbf{TTFT coverage} & Spatial placement and association \\
\bottomrule
\end{tabular}%
}
\end{table*}

\section{System Model}
The model is built in four steps. We first describe where requests and edge GPUs are located, then represent reusable prompt states, next translate caching into communication and GPU-resource costs, and finally define the queue-aware association rule. This order is important: the routing decision is meaningful only after both the computational gain and the resource cost of a cache profile have been specified.

\subsection{Spatial Edge-LLM Network and Request Model}
Edge GPU nodes form a homogeneous PPP
\begin{equation}
\Phi_{\rm b}=\{x_i\}\subset\R^2,\qquad \Phi_{\rm b}\sim\PPP(\lambda_{\rm b}),
\end{equation}
where $\lambda_{\rm b}$ is the edge-node density. Users form an independent stationary point process with density $\lambda_{\rm u}$. Each user generates requests according to a Poisson process with rate $\nu$, giving the spatial request intensity
\begin{equation}
\Lambda_{\rm u}=\lambda_{\rm u}\nu.
\end{equation}
By Slivnyak's theorem, a typical request is placed at the origin.

\begin{table}[t]
\caption{Frequently Used Notation}
\label{tab:notation}
\centering
\scriptsize
\begin{tabular}{ll}
\toprule
Symbol & Meaning \\
\midrule
$\lambda_{\rm b},\Lambda_{\rm u}$ & Edge-node density and spatial request intensity \\
$\mathcal{T},\mathcal{Q},\mathcal{S}$ & Prefix forest, request types, and cache profiles \\
$L_q,H_{s,q}$ & Input length and reusable-prefix length \\
$\pi_s,\lambda_s$ & Profile probability and profile-tier density \\
$B_s,c_s$ & Persistent memory and runtime concurrency \\
$T^{\rm com},T^{\rm pf}$ & Communication and residual prefill delay \\
$A_{s,q},R_s$ & Association probability and nearest tier distance \\
$\lambda_s^{\rm q},\mu_s,\rho_s$ & Queue arrival rate, service rate, and utilization \\
$P_{\rm cov}^{\rm TTFT}$ & TTFT coverage probability \\
\bottomrule
\end{tabular}
\end{table}

\subsection{Prefix Forest and Spatial Cache Profiles}
Let $\mathcal{T}=(\mathcal{V},\mathcal{E})$ be a rooted prefix forest. A vertex $v\in\mathcal{V}$ represents an exact reusable token prefix, with depth $\ell_v$ tokens. A request type $q\in\mathcal{Q}$ occurs with probability $p_q$, has total input length $L_q$, and follows a root-to-leaf path $\mathcal{P}_q\subseteq\mathcal{V}$. The request distribution satisfies $\sum_qp_q=1$.

A cache profile $s\in\mathcal{S}$ is an ancestor-closed subset $\mathcal{C}_s\subseteq\mathcal{V}$: if $v\in\mathcal{C}_s$, every ancestor of $v$ is also in $\mathcal{C}_s$. The longest reusable prefix provided by profile $s$ to request $q$ is
\begin{equation}
H_{s,q}=\max\left(\{\ell_v:v\in\mathcal{C}_s\cap\mathcal{P}_q\}\cup\{0\}\right).
\label{eq:hitdepth}
\end{equation}
This definition supports multiple prompt families, branching prefixes, and partial matches. A chain of nested prefixes is a special case.

Let $B_s$ denote the persistent GPU memory occupied by profile $s$. If each cached token requires $\kappa_v$ bytes at vertex $v$, then a generic profile cost is
\begin{equation}
B_s=\sum_{v\in\mathcal{C}_s\setminus\mathcal{A}_s}\kappa_v(\ell_v-\ell_{{\rm par}(v)}),
\end{equation}
where $\mathcal{A}_s$ removes duplicate accounting of shared ancestors. In the uniform-token baseline, $B_s=\kappa$ times the number of unique cached tokens.

A cache profile is intentionally more general than a single cached leaf. It may represent, for example, a globally replicated system prompt together with one popular tool description, or a retrieval template together with several high-frequency document prefixes. Ancestor closure prevents inconsistent states: a node cannot store a child KV state without retaining the state needed to reach that child. In an implementation, profiles can be generated from a trace-derived prefix forest by retaining a small set of frequent or high-value subtrees.

For request $q$, define the raw prefill saving of profile $s$ as
\begin{equation}
V_{s,q}^{\rm pf}=g(L_q)-g(L_q-H_{s,q}).
\label{eq:rawsaving}
\end{equation}
This quantity separates the semantic popularity of a prefix from its computation value. Two prefixes with the same request probability can have very different value if their depths differ substantially.

The prefix forest specifies the logical reuse relation. To analyze a network rather than a single server, that logical structure must be mapped to a spatial placement process.

Each edge node independently selects one profile $S_x\in\mathcal{S}$ with
\begin{equation}
\Prb(S_x=s)=\pi_s,\qquad \sum_{s\in\mathcal{S}}\pi_s=1.
\end{equation}

\begin{lemma}[Profile thinning]
The nodes using profile $s$ form an independent PPP
\begin{equation}
\Phi_s\sim\PPP(\lambda_s),\qquad \lambda_s=\lambda_{\rm b}\pi_s.
\end{equation}
\end{lemma}
\begin{proof}
Let $\mathcal B\subset\R^2$ be a bounded Borel set. Conditional on $N_{\rm b}(\mathcal B)=n$ original edge nodes in $\mathcal B$, the number assigned profile $s$ is binomial with parameters $(n,\pi_s)$ because marks are independent across nodes. Averaging over the Poisson law of $N_{\rm b}(\mathcal B)$ gives a Poisson random variable with mean $\lambda_{\rm b}\pi_s|\mathcal B|$. Applying the same argument jointly to disjoint spatial sets and to the multinomial profile marks shows that the resulting counting measures have independent increments and that different profile processes are mutually independent. Hence $\Phi_s$ is a homogeneous PPP with intensity $\lambda_s=\lambda_{\rm b}\pi_s$.
\end{proof}

Let $R_s$ be the distance from the typical user to the nearest profile-$s$ node. For $\lambda_s>0$,
\begin{equation}
f_{R_s}(r)=2\pi\lambda_sr\exp(-\pi\lambda_sr^2),\quad r\ge0.
\label{eq:nearest}
\end{equation}

\subsection{Communication, Computation, and GPU Memory}
We use the delay-equivalent radio model
\begin{equation}
T^{\rm com}(r)=\xi r^{\alpha},\qquad \alpha>2,
\label{eq:commdelay}
\end{equation}
where $\xi$ absorbs prompt size, bandwidth, spectral efficiency, retransmissions, and protocol overhead. This model preserves tractable association regions and is widely used as a distance-dependent service-cost abstraction. Section~\ref{sec:extensions} outlines an instantaneous SINR extension.

Let $g(n)$ denote the prefill time for $n$ uncached tokens. We assume that $g(0)=0$ and $g(n)$ is nondecreasing. A calibrated specialization is
\begin{equation}
g(n)=a n+b n^2,
\label{eq:prefill}
\end{equation}
where the linear term captures projection and memory operations and the quadratic term approximates attention computation. The residual prefill delay is
\begin{equation}
T^{\rm pf}_{s,q}=g(L_q-H_{s,q}).
\label{eq:residualprefill}
\end{equation}

Communication and residual prefill are only two parts of the cost. Persistent states also consume the same accelerator memory needed by active requests, creating the central reuse--concurrency coupling.

Each node has usable GPU memory $M$. Persistent profile $s$ leaves $M-B_s$ bytes for active requests. The serving engine applies admission control with a per-sequence runtime budget $m^{\rm run}$, yielding concurrency
\begin{equation}
c_s=\left\lfloor\frac{M-B_s}{m^{\rm run}}\right\rfloor.
\label{eq:concurrency}
\end{equation}
Profiles with $c_s<1$ are infeasible.

\subsection{Queueing and Load-Aware Association}
The mean GPU occupancy of a type-$q$ request served by profile $s$ is
\begin{equation}
S_{s,q}=T^{\rm pf}_{s,q}+D_q,
\label{eq:occupancy}
\end{equation}
where $D_q$ summarizes decode occupancy and other model-execution overhead. We use an analytically calibrated $M/M/c_s$ approximation. The approximation is not intended to reproduce every detail of continuous batching; rather, it exposes the first-order interaction between cache affinity, concurrency, and queueing. Discrete-event simulation can replace the queue module without changing the spatial derivation.

Let $A_{s,q}$ be the association probability of type $q$ to profile $s$. The arrival rate per profile-$s$ node is
\begin{equation}
\lambda^{\rm q}_s=\frac{\Lambda_{\rm u}}{\lambda_s}\sum_qp_qA_{s,q}.
\label{eq:arrival}
\end{equation}
The request mixture among profile-$s$ jobs is
\begin{equation}
\omega_{s,q}=\frac{p_qA_{s,q}}{\sum_jp_jA_{s,j}},
\end{equation}
and the effective service rate is
\begin{equation}
\mu_s=\left(\sum_q\omega_{s,q}S_{s,q}\right)^{-1}.
\label{eq:service}
\end{equation}
Define utilization $\rho_s=\lambda^{\rm q}_s/(c_s\mu_s)$. Stability requires $\rho_s<1$.

For a stable $M/M/c_s$ queue, the Erlang-C waiting probability is
\begin{equation}
P^{\rm wait}_s=
\frac{\frac{a_s^{c_s}}{c_s!(1-\rho_s)}}
{\sum_{n=0}^{c_s-1}\frac{a_s^n}{n!}+\frac{a_s^{c_s}}{c_s!(1-\rho_s)}},
\quad a_s=\frac{\lambda^{\rm q}_s}{\mu_s},
\label{eq:erlangc}
\end{equation}
and the mean waiting time is
\begin{equation}
\bar W_s=\frac{P^{\rm wait}_s}{c_s\mu_s-\lambda^{\rm q}_s}.
\label{eq:meanwait}
\end{equation}
The waiting-time CDF is
\begin{equation}
F_{W_s}(w)=1-P^{\rm wait}_s\exp[-(c_s\mu_s-\lambda^{\rm q}_s)w],\quad w\ge0.
\label{eq:waitcdf}
\end{equation}

The queueing model closes the resource side of the problem. It remains to specify how a request uses these quantities when comparing candidate profiles.

A type-$q$ request evaluates the nearest node of each nonempty profile. Its predicted TTFT cost is
\begin{equation}
J_{s,q}(r)=\xi r^{\alpha}+\bar W_s+T^{\rm pf}_{s,q}.
\label{eq:cost}
\end{equation}
The selected profile is
\begin{equation}
s_q^*=\argmin_{s:\lambda_s>0}J_{s,q}(R_s).
\label{eq:association}
\end{equation}
Association uses mean waiting time, which can be broadcast by edge nodes or estimated by a controller. Performance analysis below retains the random waiting-time distribution.

The baseline analysis makes five deliberate abstractions. First, exact prefix equality is assumed, so cached states do not change model outputs. Second, profile marks are independent across nodes; correlated or coordinated placements are discussed later. Third, the association controller uses a slowly varying communication-cost model and advertised mean queue delay, while instantaneous queue randomness is retained in coverage analysis. Fourth, the $M/M/c$ model represents admission-controlled concurrent sequences rather than individual GPU cores. Fifth, one model replica is assumed at every edge node. These assumptions expose the spatial cache--compute tradeoff with tractable mathematics; none is required by the cache-profile definition itself.

The analysis applies on a time scale over which the prefix-popularity distribution and cache profiles are approximately stationary. Faster cache admission and eviction can operate within each profile, whereas the spatial placement probabilities represent a slower planning layer.

\section{Spatial and Queue-Coupled Performance Analysis}
The analysis follows the operational sequence of a request. We first determine which profile wins the spatial competition and how far away the selected node is. We then account for the fact that these associations create the very queue delays used by the router. Once this fixed point is established, TTFT reliability and computation reuse can be evaluated, and the resulting expressions can be translated into structural design rules.

\subsection{Spatial Association Law}
Define
\begin{equation}
\Delta_{s,q}=\bar W_s+T^{\rm pf}_{s,q}
\end{equation}
and
\begin{equation}
G_{s,t,q}(r)=\pos{r^{\alpha}+\frac{\Delta_{s,q}-\Delta_{t,q}}{\xi}}^{1/\alpha}.
\label{eq:guard}
\end{equation}

\begin{theorem}[Profile-association probability]
For a type-$q$ request, the probability of associating with profile $s$ is
\begin{equation}
A_{s,q}=\int_0^\infty 2\pi\lambda_sr
\exp\!\left[-\pi\lambda_sr^2-\pi\sum_{t\ne s}\lambda_tG_{s,t,q}^2(r)\right]\,dr.
\label{eq:assocprob}
\end{equation}
\end{theorem}
\begin{proof}
Condition on $R_s=r$. The candidate profile-$s$ node wins only if its predicted cost is no larger than that of the nearest candidate in every other active profile $t$. Thus,
\begin{equation}
\xi r^\alpha+\Delta_{s,q}\le \xi R_t^\alpha+\Delta_{t,q},\qquad \forall t\ne s.
\end{equation}
After rearrangement and enforcement of the nonnegativity of distance, the event above is equivalent to $R_t\ge G_{s,t,q}(r)$ for every $t\ne s$. Because the independently marked profile processes are mutually independent, these events factor across $t$. The nearest-neighbor void probability of tier $t$ is
\begin{equation}
\Prb[R_t\ge u]=\exp(-\pi\lambda_tu^2).
\end{equation}
Therefore, the conditional probability that profile $s$ wins is the product of the corresponding void probabilities. Multiplying this product by the density of $R_s$ in~\eqref{eq:nearest} and integrating over $r$ yields~\eqref{eq:assocprob}. Ties have probability zero under the continuous distance distributions. Appendix~A provides the complete conditioning argument and normalization step.
\end{proof}

\begin{corollary}[Conditional serving distance]
The conditional density of serving distance, given profile $s$ and request type $q$, is
\begin{equation}
f_{R|s,q}(r)=\frac{f^{\rm sel}_{s,q}(r)}{A_{s,q}},
\end{equation}
where
\begin{equation}
f^{\rm sel}_{s,q}(r)=2\pi\lambda_sr
\exp\!\left[-\pi\lambda_sr^2-\pi\sum_{t\ne s}\lambda_tG_{s,t,q}^2(r)\right].
\label{eq:fsel}
\end{equation}
\end{corollary}
\begin{proof}
The numerator is the joint density of the events $\{s_q^*=s\}$ and $\{R_s\in dr\}$ obtained in the proof of Theorem~1. Dividing this joint density by $\Prb(s_q^*=s)=A_{s,q}$ is Bayes' rule for densities and gives the stated conditional law. Its integral is one because the numerator integrates to $A_{s,q}$.
\end{proof}

The cache profile acts as a computation-induced tier bias. Unlike power-biased cellular association, the bias depends on request type through prefix overlap and residual prefill.

The association law also yields the serving-distance distribution and several consistency checks. These consequences are grouped here because they all describe the geometry induced by the same request-dependent tier bias.

\begin{corollary}[Equal non-spatial costs]
If $\Delta_{s,q}=\Delta_q$ for every active profile, association reduces to nearest-node association over the superposed PPP and
\begin{equation}
A_{s,q}=\frac{\lambda_s}{\sum_t\lambda_t}=\pi_s.
\label{eq:equalcost}
\end{equation}
\end{corollary}
\begin{proof}
Equal non-spatial costs imply $G_{s,t,q}(r)=r$ for all active $s$ and $t$. Substitution into~\eqref{eq:assocprob} gives
\begin{equation}
A_{s,q}=\int_0^\infty 2\pi\lambda_s r\exp\!\left(-\pi r^2\sum_t\lambda_t\right)dr.
\end{equation}
Factoring $\lambda_s/\sum_t\lambda_t$ leaves the nearest-neighbor density of the superposed PPP, whose integral is one. Since $\sum_t\lambda_t=\lambda_{\rm b}$, the result is $A_{s,q}=\lambda_s/\lambda_{\rm b}=\pi_s$.
\end{proof}

\begin{proposition}[Cost monotonicity]
Holding tier densities and competing costs fixed, $A_{s,q}$ is nonincreasing in $\Delta_{s,q}$ and nondecreasing in any reduction of its residual prefill time.
\end{proposition}
\begin{proof}
Fix the tier densities, the costs of the competing profiles, and a distance $r$. The mapping $x\mapsto[x]_+^{1/\alpha}$ is nondecreasing, hence each guard radius $G_{s,t,q}(r)$ is nondecreasing in $\Delta_{s,q}$. Each factor $\exp[-\pi\lambda_tG_{s,t,q}^2(r)]$ is therefore nonincreasing, while the nearest-profile-$s$ density is unchanged. The integrand in~\eqref{eq:assocprob} is thus ordered pointwise. Integration preserves the order, proving that $A_{s,q}$ is nonincreasing in $\Delta_{s,q}$. Since $T^{\rm pf}_{s,q}$ enters $\Delta_{s,q}$ additively, reducing the residual prefill time cannot reduce the association probability. Strict monotonicity holds whenever a positive-measure set of distances has an active competing guard radius and $\lambda_t>0$ for at least one competitor.
\end{proof}

The conditional mean serving distance is
\begin{equation}
\bar R_{s,q}=\frac{1}{A_{s,q}}\int_0^\infty r f^{\rm sel}_{s,q}(r)\\,dr,
\label{eq:meanservingdistance}
\end{equation}
and the network-wide mean distance is $\bar R=\sum_qp_q\sum_sA_{s,q}\bar R_{s,q}$. Unlike nearest-node routing, $\bar R$ may increase when prefix reuse becomes more valuable. This is not an inefficiency by itself: the additional radio path can be justified by a larger computation saving.

\subsection{Load--Association Coupling and Stability}
The preceding expressions treat the advertised queue delays as fixed. In operation, however, the association probabilities determine the arrival rate of each profile tier, and those arrivals determine its delay. The spatial and queueing models must therefore be solved jointly.

Association probabilities depend on $\bar W_s$, while queue delay depends on the arrivals generated by those probabilities. Let $\mathcal{A}(\bar{\mathbf W})$ denote~\eqref{eq:assocprob} and let $\mathcal{Q}(\mathbf A)$ denote~\eqref{eq:arrival}--\eqref{eq:meanwait}. A consistent operating point satisfies
\begin{equation}
\bar{\mathbf W}^{*}=\mathcal{Q}(\mathcal{A}(\bar{\mathbf W}^{*})).
\label{eq:fixedpoint}
\end{equation}
We use the damped iteration
\begin{equation}
\bar{\mathbf W}^{(i+1)}=(1-\eta)\bar{\mathbf W}^{(i)}+\eta\mathcal{Q}(\mathcal{A}(\bar{\mathbf W}^{(i)})),
\label{eq:damped}
\end{equation}
with $0<\eta\le1$. A placement is infeasible if any active profile reaches $\rho_s\ge1$.

\begin{remark}
The fixed point implements a negative-feedback mechanism. If a deep-prefix tier attracts excessive demand, its queue delay rises, reducing its future association region. Queue-unaware routing removes this feedback and can create a cache-affinity hotspot.
\end{remark}

A fixed point is physically meaningful only when every active tier can serve the traffic assigned to it. The following capacity condition and implementation rule make this stability requirement explicit.

A necessary network-wide stability condition is obtained by comparing offered traffic with aggregate service capacity per unit area.

\begin{proposition}[Area-capacity necessary condition]
For a fixed request mixture and active profile set, any stable association must satisfy
\begin{equation}
\Lambda_{\rm u}<\sum_{s\in\mathcal{S}}\lambda_s c_s\mu_s.
\label{eq:areastability}
\end{equation}
\end{proposition}
\begin{proof}
Multiplying the per-node arrival rate in~\eqref{eq:arrival} by the spatial node density and summing over profiles gives
\begin{equation}
\sum_s\lambda_s\lambda_s^{\rm q}
=\Lambda_{\rm u}\sum_qp_q\sum_sA_{s,q}=\Lambda_{\rm u},
\end{equation}
where the last equality follows because every request is associated with exactly one active profile. The maximum nominal service rate contributed by profile $s$ per unit area is $\lambda_sc_s\mu_s$. If the total offered arrival intensity is at least the sum of these capacities, rate conservation precludes all profile queues from having a strictly smaller arrival rate than service rate. Consequently, at least one active tier must violate $\rho_s<1$, establishing the necessary condition. The condition is not sufficient because cache affinity may overload a sparse tier even when aggregate capacity is adequate.
\end{proof}

Condition~\eqref{eq:areastability} is necessary but not sufficient because cache affinity may concentrate traffic on a sparse profile even when total capacity is adequate. The fixed point resolves this imbalance by enlarging $\bar W_s$ at an overloaded tier, which shrinks its association region. In implementation, we initialize all delays at zero, apply damping, and stop when
\begin{equation}
\left\|\bar{\mathbf W}^{(i+1)}-\bar{\mathbf W}^{(i)}\right\|_\infty<\epsilon_{\rm fp}.
\end{equation}
When a candidate profile becomes unstable, its advertised cost is set to a large penalty during the outer placement search. Fig.~\ref{fig:fixedpoint} later shows stable convergence for the baseline workload.

\subsection{TTFT Reliability and Computation Reuse}
After the load--association fixed point has been obtained, the same joint selection-distance density can be combined with the queue-delay distribution to characterize user-perceived latency.

The realized TTFT of a request served at profile $s$ and distance $r$ is
\begin{equation}
T_{s,q}=\xi r^\alpha+W_s+T^{\rm pf}_{s,q}.
\end{equation}

\begin{theorem}[TTFT coverage]
For latency target $\tau$, the network-wide TTFT coverage probability is
\begin{align}
P_{\rm cov}^{\rm TTFT}(\tau)
=&\sum_qp_q\sum_s\int_0^{r^{\max}_{s,q}(\tau)} f^{\rm sel}_{s,q}(r)\nonumber\\
&\times F_{W_s}(\tau-\xi r^\alpha-T^{\rm pf}_{s,q})\,dr,
\label{eq:ttftcoverage}
\end{align}
where
\begin{equation}
r^{\max}_{s,q}(\tau)=\pos{\frac{\tau-T^{\rm pf}_{s,q}}{\xi}}^{1/\alpha}.
\end{equation}
\end{theorem}
\begin{proof}
Condition on request type $q$, selected profile $s$, and serving distance $r$. The event $T_{s,q}\le\tau$ is equivalent to
\begin{equation}
W_s\le\tau-\xi r^\alpha-T^{\rm pf}_{s,q}.
\end{equation}
The right-hand side is nonnegative only for $r\le r^{\max}_{s,q}(\tau)$. Multiplying the waiting-time CDF by the joint selection-distance density~\eqref{eq:fsel}, integrating, and averaging over $q$ and $s$ proves the result. Appendix~B expands the conditioning steps and the truncation radius argument.
\end{proof}

Equation~\eqref{eq:ttftcoverage} is the LLM-inference analogue of wireless coverage probability. It simultaneously includes spatial availability, partial computation reuse, and queue reliability.

The coverage probability is the primary service-level metric, while the mean, quantile, and token-level reuse metrics reveal which component of latency changes and how much computation is actually avoided.

The mean TTFT follows by conditioning on request type, selected profile, and serving distance:
\begin{align}
\bar T^{\rm TTFT}
=&\sum_qp_q\sum_s\int_0^\infty f^{\rm sel}_{s,q}(r)\nonumber\\
&\times\left(\xi r^\alpha+\bar W_s+T^{\rm pf}_{s,q}\right)dr.
\label{eq:meanttft}
\end{align}
Equation~\eqref{eq:meanttft} uses the same mean delay employed for association. The coverage expression~\eqref{eq:ttftcoverage} is more informative for service-level objectives because it retains the waiting-time tail.

A queue-free upper bound is obtained by replacing $F_{W_s}(w)$ with one in~\eqref{eq:ttftcoverage}. A conservative no-wait lower approximation for planning can instead reduce the target by a chosen queue-delay quantile. These bounds separate failures caused by spatial communication and prefill from those caused by congestion.

For reliability target $\zeta$, define the TTFT quantile
\begin{equation}
\tau_\zeta=\inf\\{\tau:P_{\rm cov}^{\rm TTFT}(\tau)\ge\zeta\\}.
\label{eq:ttftquantile}
\end{equation}
This quantity is directly compatible with percentile-based serving objectives.

The expected reusable-prefix length is
\begin{equation}
\bar H=\sum_qp_q\sum_sA_{s,q}H_{s,q}.
\end{equation}
We define the normalized computation-offloading ratio
\begin{equation}
\eta_{\rm off}=\sum_qp_q\sum_sA_{s,q}\frac{H_{s,q}}{L_q}.
\label{eq:offloading}
\end{equation}
Unlike a binary file-hit metric, $\eta_{\rm off}$ measures the fraction of prompt computation avoided by the selected spatial cache.

\subsection{Structural Design Insights}
The previous results can be summarized by a simple pairwise comparison: a longer path is justified only when the non-spatial latency saving is larger than the additional communication cost.

\begin{proposition}[Remote-prefix selection condition]
Consider candidate profiles $s$ and $t$ at distances $r_s$ and $r_t$. Profile $t$ is preferred to profile $s$ if and only if
\begin{equation}
\xi(r_t^\alpha-r_s^\alpha)<\Delta_{s,q}-\Delta_{t,q}.
\label{eq:remotecondition}
\end{equation}
\end{proposition}
\begin{proof}
Profile $t$ is selected over $s$ exactly when
\begin{equation}
\xi r_t^\alpha+\Delta_{t,q}<\xi r_s^\alpha+\Delta_{s,q}.
\end{equation}
Moving the communication terms to the left and the non-spatial terms to the right gives~\eqref{eq:remotecondition}. Equality defines the association boundary. Expanding $\Delta_{s,q}-\Delta_{t,q}$ further separates the prefill saving of $t$ from its queue-delay disadvantage, which explains why a deeper but congested cache may cease to be attractive.
\end{proof}

When $t$ provides a deeper prefix, the right-hand side includes its prefill saving, offset by any additional queueing delay. Long prompts and expensive prefill enlarge the region in which a remote cache is beneficial.

Fig.~\ref{fig:remoteregion} visualizes~\eqref{eq:remotecondition} for two candidate nodes. The boundary is curved because radio delay is superlinear in distance. A deep-prefix node can remain preferable even when it is substantially farther away, but the admissible distance advantage collapses as queue delay at that node rises.

\begin{figure}[t]
\centering
\includegraphics[width=\columnwidth]{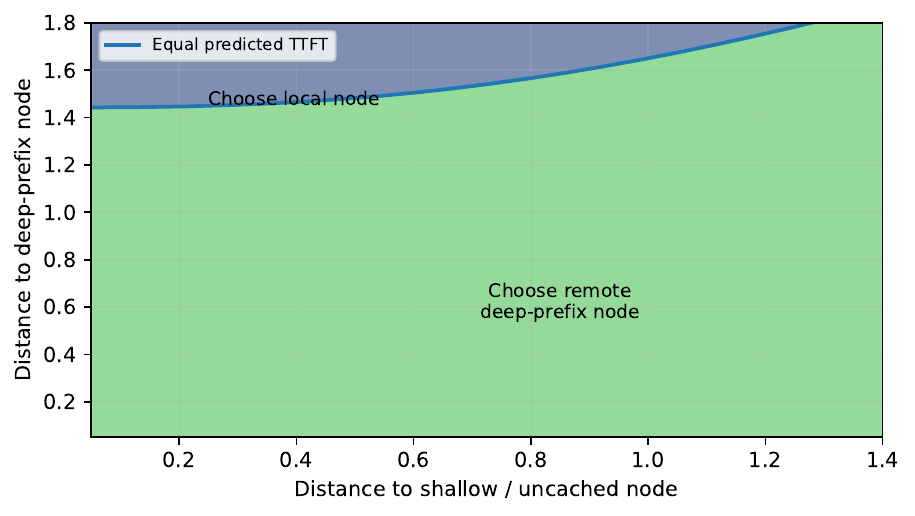}
\caption{Illustrative selection region for a local shallow-cache node and a remote deep-prefix node. The remote node is selected when its non-spatial latency saving exceeds the additional communication delay.}
\label{fig:remoteregion}
\end{figure}

This condition links the entire analysis. A deeper prefix lowers residual prefill, but its profile may be sparse and congested; replicating it reduces distance and load, but consumes more persistent memory; preserving uncached or shallow profiles sacrifices some reuse while protecting concurrency. The optimal design is therefore hierarchical rather than uniformly cache-maximizing: broad shallow reuse, selective deep specialization, and sufficient high-concurrency fallback capacity.
\section{Cache-Profile Construction and Placement Optimization}
The analytical model assumes a finite set of cache profiles, whereas a real workload produces a much larger prefix forest. This section first compresses that forest into a tractable candidate set, then optimizes how frequently each profile appears in space, and finally describes the nested numerical procedure used to evaluate a candidate placement.

\subsection{Trace-Driven Cache-Profile Construction}
The profile set $\mathcal{S}$ should be much smaller than the power set of prefix vertices. We use a three-stage construction procedure.

\begin{enumerate}
\item \emph{Trace aggregation:} tokenize requests with the deployed model tokenizer and insert them into a radix forest. Each vertex records request count, reusable depth, state size, and lifetime.
\item \emph{Value scoring:} assign vertex $v$ a score such as
\begin{equation}
\Gamma_v=\frac{\widehat p_v\\,[g(L_v)-g(L_v-\ell_v)]}{B_v},
\label{eq:profilescore}
\end{equation}
which estimates saved prefill per persistent byte.
\item \emph{Profile generation:} create a limited number of ancestor-closed subtrees using greedy expansion, clustering by request family, or a tree knapsack algorithm. Include an uncached/high-concurrency profile as a fallback.
\end{enumerate}

A profile set can be recomputed at a slow time scale, while the placement probabilities are optimized at an intermediate time scale and request routing operates online.

\begin{lemma}[Dominated-profile pruning]
Under node-aware routing, profile $s$ is weakly dominated by profile $t$ if $B_t\le B_s$, $H_{t,q}\ge H_{s,q}$ for all $q$, and $c_t\ge c_s$. Removing $s$ from the candidate configuration set cannot decrease the optimum of~\eqref{prob:p1}.
\end{lemma}
\begin{proof}
Take any feasible deployment and routing policy that uses profile $s$. Reconfigure each profile-$s$ node with profile $t$ while retaining the physical node identity. The reconfiguration consumes no additional persistent memory, preserves at least the same concurrency, and provides no smaller reusable depth for any request type. A node-aware controller can reproduce every routing decision of the original policy by sending a request to the same physical node and, if desired, ignoring the extra cached state. Therefore the transformed deployment has no larger communication distance, residual prefill delay, or queue workload than the original benchmark. It is feasible and achieves an objective value no smaller. Hence some optimum assigns zero probability to the dominated configuration. Appendix~D gives the constructive argument in detail.
\end{proof}

\subsection{Joint Spatial Placement Problem}
The candidate profiles define what can be stored; the remaining decision is how densely each profile should be deployed. Because placement changes both spatial availability and queue capacity, the objective must be evaluated at the load--association fixed point.

Let $\boldsymbol{\pi}=(\pi_s:s\in\mathcal{S})$. For target $\tau_0$, we optimize
\begin{subequations}\label{prob:p1}
\begin{align}
\max_{\boldsymbol{\pi}}\quad &P_{\rm cov}^{\rm TTFT}(\tau_0)\\
\mathrm{s.t.}\quad &\sum_s\pi_s=1,\quad \pi_s\ge0,\\
&\sum_s\pi_sB_s\le\bar B,\\
&B_s<M,\quad \forall s:\pi_s>0,\\
&\rho_s(\boldsymbol{\pi})<1,\quad \forall s:\pi_s>0.
\end{align}
\end{subequations}
The average budget $\bar B$ can model a fleet-wide memory reservation or a cost constraint. The problem is nonconvex because placement changes spatial density, request association, request mixture, service rate, concurrency, and queue stability.

\subsection{Load-Aware Solution Algorithm and Network Planning}
Algorithm~\ref{alg:solver} uses an inner load-association fixed point and an outer projected simultaneous-perturbation search. Other derivative-free methods can replace the outer loop.

\begin{algorithm}[t]
\caption{Load-Aware Cache-Profile Placement}
\label{alg:solver}
\begin{algorithmic}[1]
\Require Profile set $\mathcal{S}$, workload $\{p_q,L_q\}$, memory $M$, budget $\bar B$, target $\tau_0$.
\State Initialize feasible $\boldsymbol{\pi}^{(0)}$.
\For{$n=0,1,\ldots$}
  \State Set $\lambda_s=\lambda_{\rm b}\pi_s^{(n)}$ and $c_s$ using~\eqref{eq:concurrency}.
  \State Solve~\eqref{eq:fixedpoint} using~\eqref{eq:damped}.
  \If{any active profile is unstable}
     \State Assign a large penalty to the objective.
  \Else
     \State Evaluate~\eqref{eq:ttftcoverage}.
  \EndIf
  \State Estimate a search direction $\widehat{\nabla}P_{\rm cov}^{\rm TTFT}$ by simultaneous perturbation.
  \State $\widetilde{\boldsymbol{\pi}}\gets\boldsymbol{\pi}^{(n)}+\gamma_n\widehat{\nabla}P_{\rm cov}^{\rm TTFT}$.
  \State Project $\widetilde{\boldsymbol{\pi}}$ onto the simplex and memory-budget set to obtain $\boldsymbol{\pi}^{(n+1)}$.
  \If{objective improvement $<\epsilon$} \State \textbf{break} \EndIf
\EndFor
\end{algorithmic}
\end{algorithm}

With $N_r$ quadrature points, one fixed-point iteration requires $O(|\mathcal{Q}||\mathcal{S}|^2N_r)$ operations. Candidate-profile generation should therefore prune dominated profiles: profile $s$ is dominated if another profile has no larger memory, no smaller hit depth for any request, and no lower concurrency.

The outer search must remain inside the simplex, memory budget, and stability region. Projection and profile pruning provide a practical implementation, while the same solver can be inverted to size the edge deployment itself.

The outer feasible set is the intersection of a probability simplex and a weighted memory half-space. For a trial vector $\mathbf z$, projection solves
\begin{equation}
\min_{\boldsymbol{\pi}}\frac12\\|\boldsymbol{\pi}-\mathbf z\\|_2^2
\quad\text{s.t.}\quad
\boldsymbol{1}^{\mathsf T}\boldsymbol{\pi}=1,\\
\boldsymbol{B}^{\mathsf T}\boldsymbol{\pi}\le\bar B,\\
\boldsymbol{\pi}\succeq0.
\label{eq:projection}
\end{equation}
For tens of profiles, this small convex subproblem is negligible relative to fixed-point evaluation. For a small profile set, simplex grid search gives a reproducible global benchmark. For a large trace-derived profile set, simultaneous perturbation or Bayesian optimization can be combined with pruning and warm-started fixed points.

The framework also supports two operational constraints that are useful in practice. A minimum fallback fraction $\pi_0\ge\pi_{\min}$ preserves high-concurrency nodes, while a maximum deep-cache fraction prevents a specialized tier from receiving too little spatial density. These constraints can be added linearly to~\eqref{prob:p1}.

The same framework can solve the dual planning problem
\begin{equation}
\min_{\lambda_{\rm b},\boldsymbol{\pi}}\lambda_{\rm b}
\quad\text{s.t.}\quad
P_{\rm cov}^{\rm TTFT}(\tau_0)\ge\zeta,
\label{eq:densityplanning}
\end{equation}
which gives the minimum edge-GPU density required for a TTFT reliability target $\zeta$.

\section{Numerical Results}
The numerical study has two roles. The first is verification: the spatial integrals, queue tail, and fixed-point solver should reproduce direct simulation under the same assumptions. The second is interpretation: once verified, the model is used to isolate the communication--computation--queueing mechanisms that a cache-hit ratio alone cannot reveal.

\subsection{Setup and Analytical Validation}
We evaluate the analytical framework and verify its principal distributions by Monte Carlo or event-driven simulation. The experiments are intentionally normalized: their purpose is to establish structural laws and validate the derivations, not to claim a hardware-specific service latency. The source package includes the figure generator, quadrature routines, fixed-point solver, and random seeds. A deployment-oriented follow-up should replace $g(n)$, $B_s$, $m^{\rm run}$, and $D_q$ with measurements from a selected model, accelerator, and serving engine.

The spatial simulation does not truncate a PPP to a finite window. For each profile $s$, it samples the exact nearest-distance distribution using
\begin{equation}
R_s=\sqrt{\frac{E_s}{\pi\lambda_s}},\qquad E_s\sim\mathrm{Exp}(1),
\label{eq:directsample}
\end{equation}
which removes boundary bias. Queue validation uses an event-driven FCFS multi-server simulator with exponential interarrival and service times.

\begin{table}[t]
\caption{Baseline Numerical Parameters}
\label{tab:params}
\centering
\scriptsize
\begin{tabular}{ll}
\toprule
Parameter & Value \\
\midrule
Edge density $\lambda_{\rm b}$ & $1.2$ normalized nodes/km$^2$ \\
Request intensity $\Lambda_{\rm u}$ & $3$ requests/(km$^2\cdot$s) \\
Communication model & $\xi=0.28$, $\alpha=2.5$ \\
Request lengths & $L_1=1024$, $L_2=4096$ tokens \\
Request probabilities & $p_1=0.4$, $p_2=0.6$ \\
Profile probabilities & $(0.45,0.35,0.20)$ \\
Reusable depths & $(0,512,2048)$ tokens \\
Concurrency & $(18,14,6)$ requests \\
Prefill model & $a=1.6\times10^{-4}$, $b=1.8\times10^{-8}$ \\
Decode occupancy & $D_q=0.4$ s \\
Fixed-point damping & $\eta=0.3$ \\
\bottomrule
\end{tabular}
\end{table}

We compare: (i) nearest-node routing, which ignores prefix and queue state; (ii) prefix-aware but queue-unaware routing; (iii) the proposed load-aware policy; and, for placement experiments, (iv) fixed uniform profile mixing versus a constrained profile search.

We begin with the spatial law, then validate the queue-delay tail, and finally verify that the two components converge when coupled through the advertised delay.

Fig.~\ref{fig:validation} compares Theorem~1 with $2\times10^5$ Monte Carlo samples. The analytical and simulated bars are visually indistinguishable, with maximum absolute error below $10^{-3}$. The $4096$-token request strongly prefers the deep-prefix profile because its prefill saving justifies a longer path. The $1024$-token request distributes more traffic to shallow and uncached profiles.

\begin{figure}[t]
\centering
\includegraphics[width=\columnwidth]{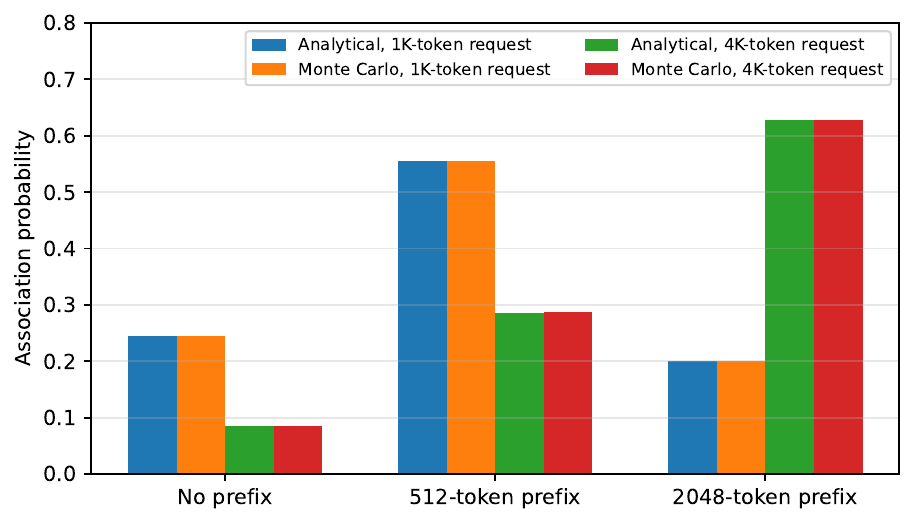}
\caption{Analytical and Monte Carlo profile-association probabilities.}
\label{fig:validation}
\end{figure}

Fig.~\ref{fig:waitcdf} compares~\eqref{eq:waitcdf} with an event-driven $M/M/9$ queue at utilization $\rho\approx0.82$. Agreement verifies the queue-tail component used in TTFT coverage. This validation does not claim that production continuous batching is exactly Markovian; it confirms that the analysis and simulator implement the same calibrated abstraction.

\begin{figure}[t]
\centering
\includegraphics[width=\columnwidth]{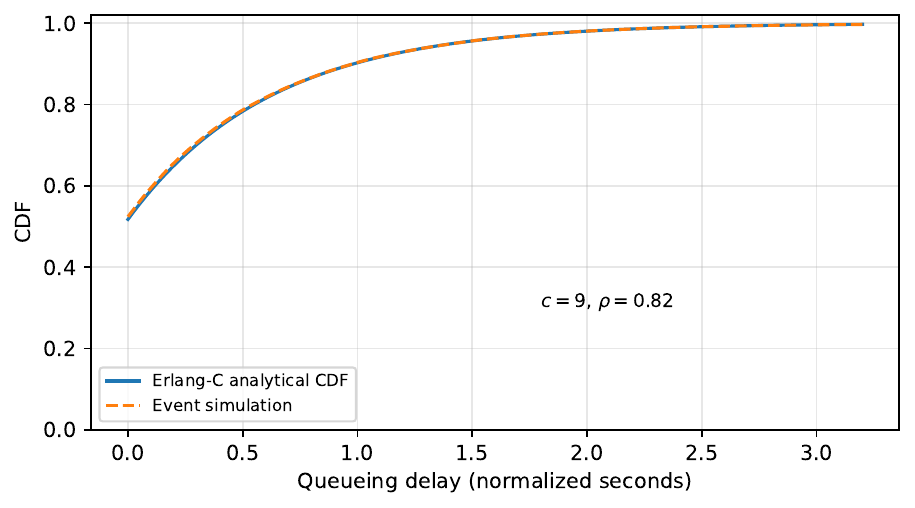}
\caption{Analytical Erlang-C waiting-time CDF and event-driven queue simulation.}
\label{fig:waitcdf}
\end{figure}

Fig.~\ref{fig:fixedpoint} shows convergence of profile-specific mean waiting times under damping. The deep-prefix tier initially attracts a large fraction of long requests and develops the largest delay. Its advertised cost then increases, shifting part of the demand to shallower profiles. The iteration settles after a small number of effective updates and remains stable under warm starts across neighboring load points.

\begin{figure}[t]
\centering
\includegraphics[width=\columnwidth]{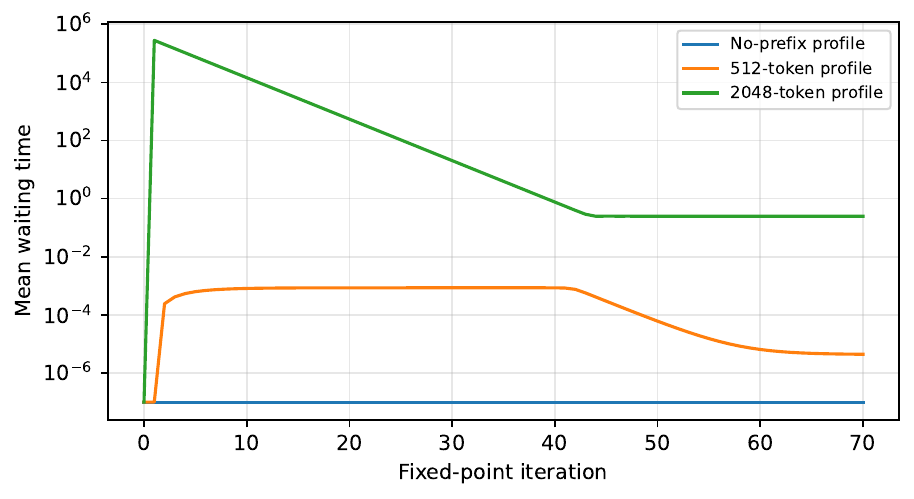}
\caption{Convergence of the damped load--association fixed point.}
\label{fig:fixedpoint}
\end{figure}

\subsection{Communication--Computation--Queueing Tradeoff}
With the analytical machinery validated, we next vary the two mechanisms that most directly change online routing: persistent prefix depth and offered load.

Fig.~\ref{fig:memory} varies a uniform cached-prefix depth and associates a decreasing concurrency with greater persistent memory. At light load, deeper reuse initially improves TTFT coverage. At moderate and heavy load, the optimum moves toward shallower caching because runtime concurrency becomes the bottleneck. Past the optimum, additional static cache rapidly increases queue delay or causes instability. This non-monotonicity is the central distinction between GPU computation-state caching and storage-only content caching.

\begin{figure}[t]
\centering
\includegraphics[width=\columnwidth]{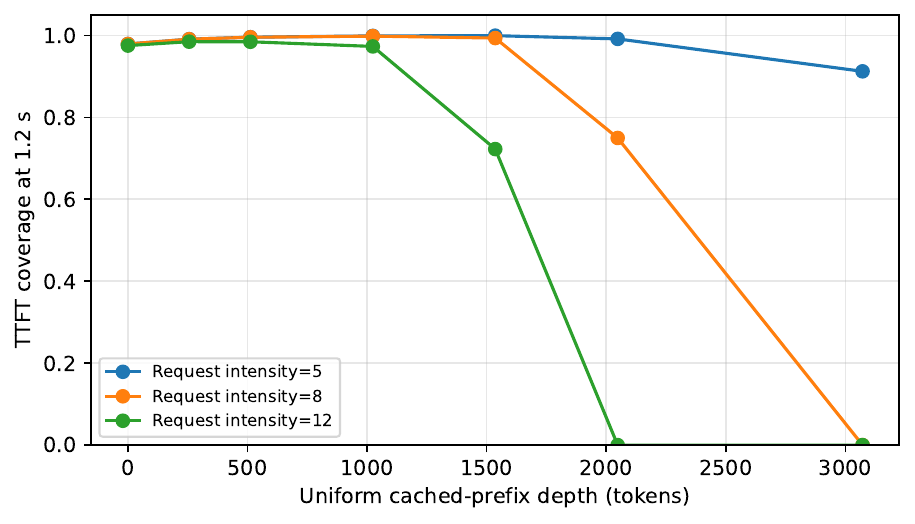}
\caption{TTFT coverage versus uniform cached-prefix depth at three traffic intensities.}
\label{fig:memory}
\end{figure}

The same memory allocation can appear beneficial at low load and harmful at high load because deeper caching changes not only prefill time but also the traffic concentration and concurrency of the selected tier.

Fig.~\ref{fig:load} plots TTFT coverage at $0.8$~s versus spatial request intensity. Under light load, queue-unaware and load-aware prefix routing are almost identical. At intensity $3$, the proposed policy achieves approximately $0.69$ coverage compared with $0.55$ for queue-unaware routing. At intensity $4$, the gap widens to approximately $0.63$ versus $0.23$. Nearest-node routing remains around $0.52$: it avoids a concentrated cache hotspot but leaves large prefill savings unused. Thus, cache awareness without load awareness can perform worse than ignoring cache affinity.

\begin{figure}[t]
\centering
\includegraphics[width=\columnwidth]{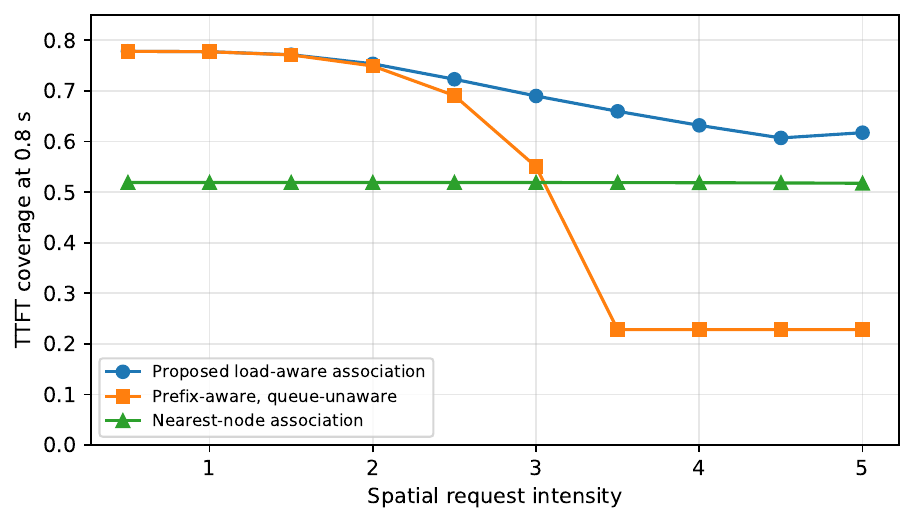}
\caption{TTFT coverage versus spatial request intensity.}
\label{fig:load}
\end{figure}

\subsection{Memory and Placement Design}
The next experiments move from per-request routing to network design. They show how workload composition changes the useful mix of shallow, deep, and fallback profiles.

Fig.~\ref{fig:popularity} increases the probability of the long request while keeping the profile mix fixed. The long request places stronger demand on the deep-prefix tier. As its popularity increases, the router expands use of high-concurrency fallback profiles to control queues. The normalized computation-offloading ratio decreases because deep-prefix capacity does not scale with demand, and mean serving distance changes as the association balance shifts. The experiment demonstrates that popularity skew affects not only cache-hit opportunity but also the load carried by each spatial tier.

\begin{figure}[t]
\centering
\includegraphics[width=\columnwidth]{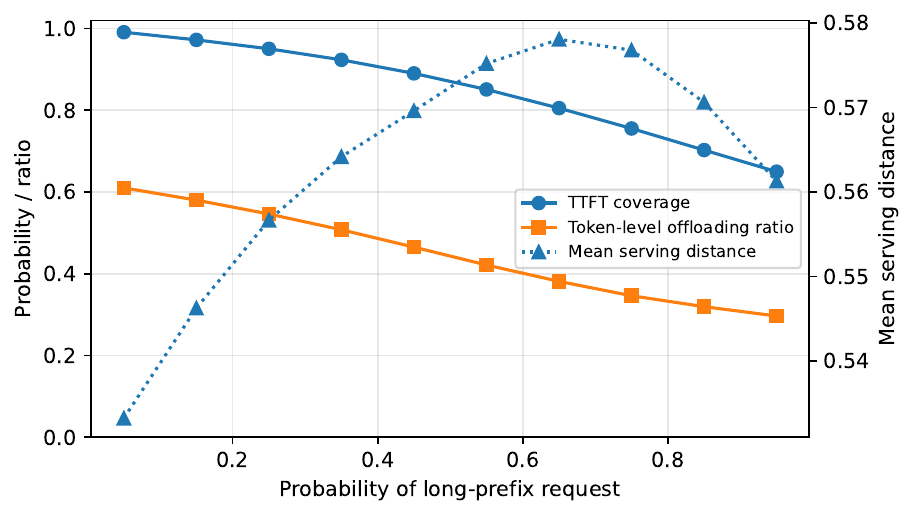}
\caption{Sensitivity to the probability of a long-prefix request.}
\label{fig:popularity}
\end{figure}

Fig.~\ref{fig:profilemix} reports the best profile mix among a feasible candidate set. At light and moderate load, the selected design retains a $20\%$ deep-prefix tier because prefill savings dominate. At high load, the selected deep-prefix fraction falls to $15\%$ and the shallow tier expands to $50\%$, preserving more concurrency while maintaining broad reuse. Fig.~\ref{fig:placementgain} shows that adapting the mix prevents the sharp degradation of a fixed placement near heavy load.

\begin{figure}[t]
\centering
\includegraphics[width=\columnwidth]{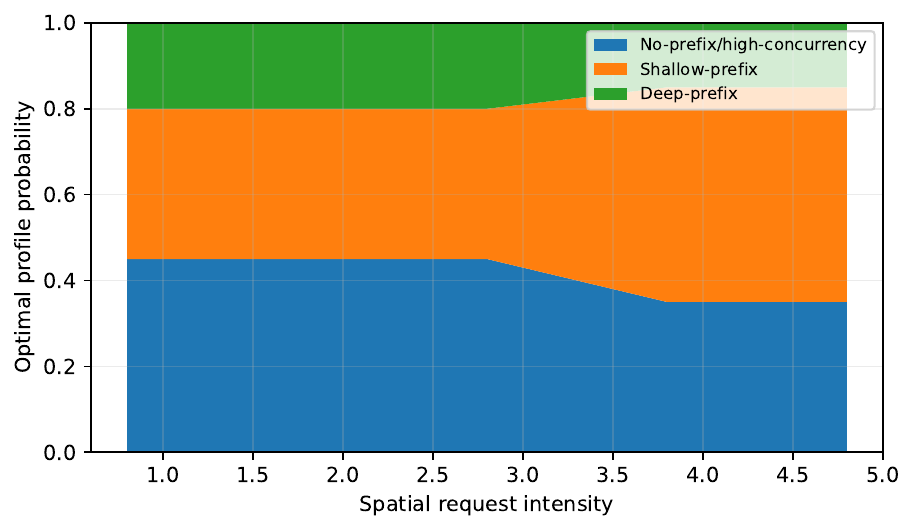}
\caption{Load-dependent cache-profile mix selected from a memory-feasible candidate set.}
\label{fig:profilemix}
\end{figure}

\begin{figure}[t]
\centering
\includegraphics[width=\columnwidth]{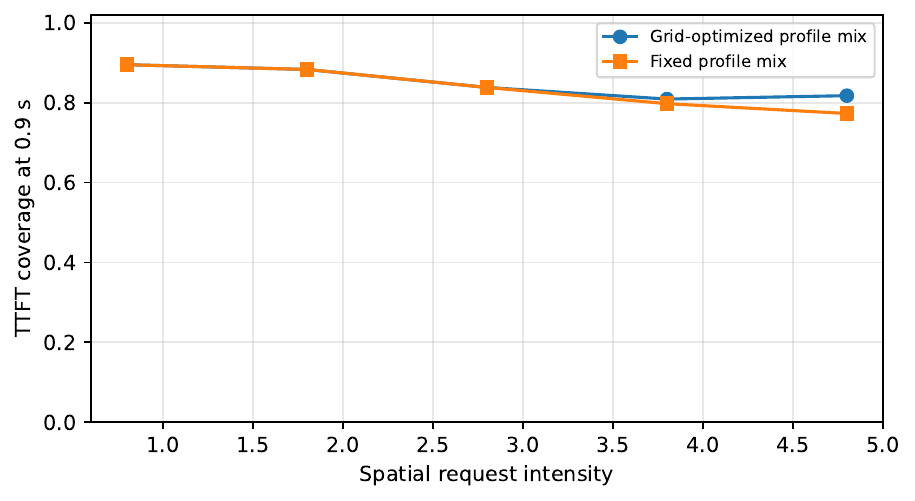}
\caption{TTFT-coverage gain from adapting the cache-profile distribution to load.}
\label{fig:placementgain}
\end{figure}

\subsection{Network Planning and Main Takeaways}
Finally, we vary edge-GPU density to separate the value of additional infrastructure from the value of cache diversity.

Fig.~\ref{fig:density} varies edge-GPU density at fixed spatial demand. Increasing density shortens communication distance, increases the spatial availability of every profile, and reduces request arrivals per node. Proposed coverage rises from approximately $0.48$ at density $0.5$ to $0.94$ at density $3$. The nearest-node baseline saturates near $0.52$ because density alone does not ensure a useful profile match. This result supports joint density and profile planning rather than treating GPU deployment and cache allocation as independent stages.

\begin{figure}[t]
\centering
\includegraphics[width=\columnwidth]{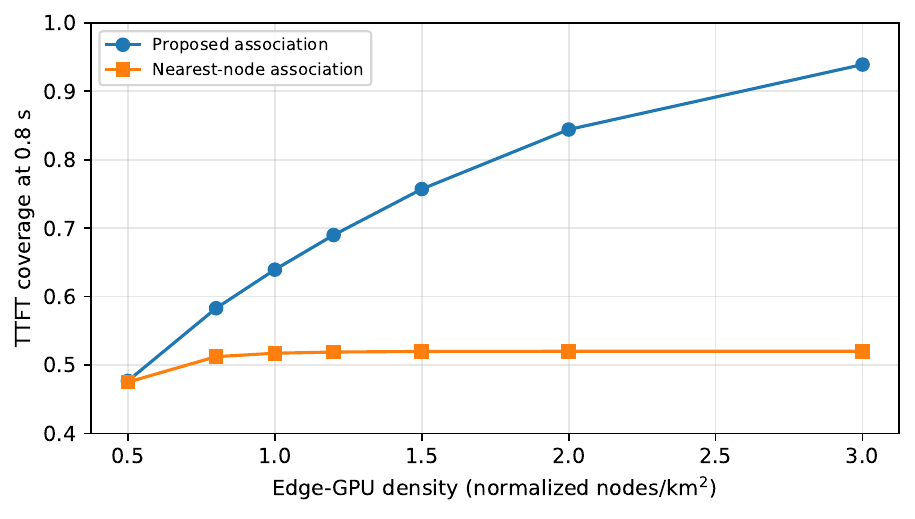}
\caption{TTFT coverage versus edge-GPU density.}
\label{fig:density}
\end{figure}

The experiments support four mechanisms rather than a single cache-hit explanation. First, computation reuse creates a request-dependent spatial bias. Second, queue feedback is necessary because cache affinity creates hotspots. Third, persistent memory has an opportunity cost in runtime concurrency. Fourth, deployment density and profile diversity interact: a specialized profile is useful only if it is spatially available often enough and provisioned with enough aggregate service capacity.

\begingroup\small
\begingroup\small
\section{Deployment Architecture and Practical Considerations}
The analytical framework is intended to guide a system rather than replace one. A practical design separates slow trace analysis and hardware calibration from periodic state placement and per-request routing. The same decomposition also clarifies which baseline assumptions can be replaced by measurements without changing the spatial logic.

\subsection{Three-Time-Scale Deployment Architecture}
The analytical model can be implemented as a hierarchical controller rather than a single monolithic optimization. Fig.~\ref{fig:deploymentloop} maps the mathematical quantities to three operational time scales. At the slowest time scale, prompt traces are tokenized, aggregated into a prefix forest, and combined with hardware profiles of prefill latency, persistent-state size, and runtime concurrency. This stage constructs and prunes the candidate profile set and can run every several hours or days. At an intermediate time scale, the controller updates regional profile probabilities, migrates selected states, publishes a cache catalogue, and recomputes the load--association fixed point. This stage responds to diurnal popularity shifts, changes in GPU availability, and measured memory pressure. At the fastest time scale, every request is classified, matched against the catalogue, and routed using the communication--prefill--queue cost in~\eqref{eq:cost}.

\begin{figure}[t]
\centering
\includegraphics[width=\columnwidth]{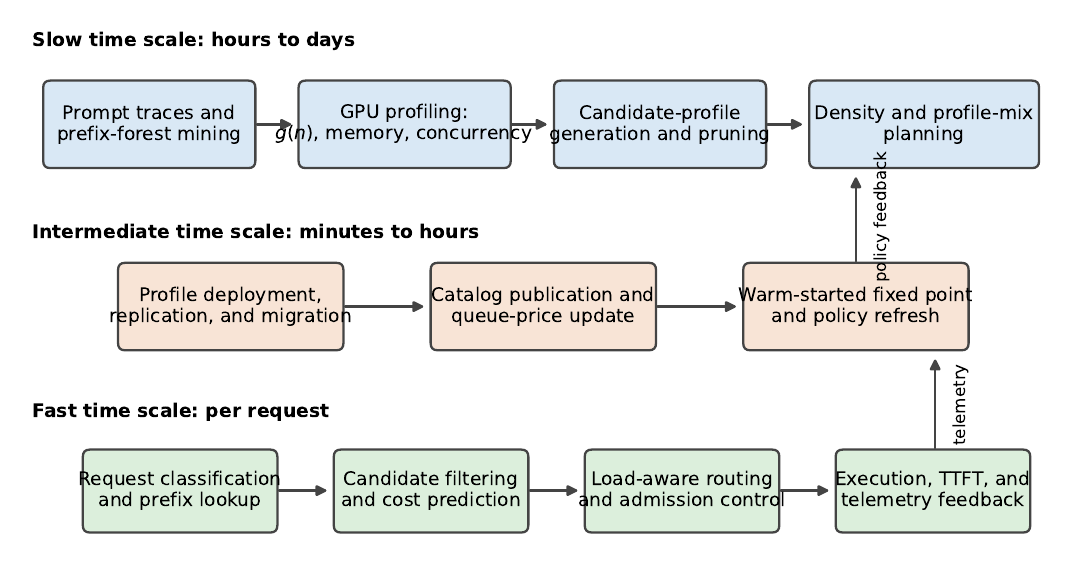}
\caption{Three-time-scale deployment architecture. Trace mining and hardware profiling operate slowly, profile placement and catalogue updates operate periodically, and association and admission decisions are made per request.}
\label{fig:deploymentloop}
\end{figure}

This separation is important for scalability. The outer profile-placement problem need not be solved for every request, while the online router need not inspect raw prompt traces or rerun GPU profiling. It only requires a compact state vector for each active profile tier. The fixed-point solution also provides a natural warm start for adjacent traffic epochs, because arrival rates and queue prices generally change gradually.

The slow, intermediate, and online loops are connected through a compact catalogue and telemetry interface. The remaining mechanisms are best understood as stages in the lifecycle of a request and a cached state rather than as independent modules.

\noindent\textbf{Online execution and robustness.}
A deployable router requires more than the cache-profile identifier. Each edge node, or a regional controller representing a group of statistically equivalent nodes, can periodically advertise
\begin{equation}
\mathcal{Z}_x=\{S_x,e_x,c_x,\widehat W_x,\widehat \rho_x,\widehat \mu_x,t_x^{\rm upd}\},
\label{eq:advertisement}
\end{equation}
where $e_x$ is the cache-catalogue epoch, $c_x$ is the admitted concurrency, $\widehat W_x$ is the estimated queue delay, $\widehat \rho_x$ is utilization, $\widehat \mu_x$ is effective service rate, and $t_x^{\rm upd}$ is the update time. The request carries only a prefix identifier or secure digest and its uncached suffix length; the router need not receive plaintext prompt content.

For node-level routing, the estimated cost can be written as
\begin{equation}
\widehat J_{x,q}=\widehat T_x^{\rm com}+g(L_q-\widehat H_{x,q})+\widehat W_x
+\widehat T_{x,q}^{\rm fetch}+\psi_x,
\label{eq:practicalcost}
\end{equation}
where $\widehat T_{x,q}^{\rm fetch}$ is zero for a local persistent state and positive for disaggregated KV retrieval. The margin $\psi_x$ accounts for stale telemetry, catalogue uncertainty, migration in progress, or prediction error. A simple robust choice is
\begin{equation}
\psi_x=z_\delta\sqrt{\sigma^2_{{\rm com},x}+\sigma^2_{W,x}+\sigma^2_{{\rm pf},x}},
\end{equation}
where $z_\delta$ is selected for a desired confidence level. This converts mean-cost routing into a risk-aware policy suitable for percentile TTFT objectives.

Candidate filtering should precede cost evaluation. Nodes with an incompatible model, adapter, tokenizer, privacy domain, or stale catalogue epoch are removed. The remaining set can be restricted to the nearest few nodes per useful profile, because a farther node with the same profile and no smaller advertised delay is dominated for that request. This reduces online complexity from network-wide search to a small candidate list.

This online decision must be backed by explicit admission, fallback, and failure-handling rules.
The fixed-point analysis assumes stable queues, but a production system must actively enforce stability. A profile tier can reject or redirect arrivals when its utilization exceeds an admission threshold $\rho_s^{\rm adm}<1$. The router then evaluates the next-best candidate rather than waiting for instability. To avoid oscillation, the threshold can use hysteresis: a tier closes at $\rho_s^{\rm high}$ and reopens at $\rho_s^{\rm low}<\rho_s^{\rm high}$.

A non-specialized fallback profile is operationally valuable even when it is not optimal under light load. It provides four protections: it preserves high runtime concurrency, serves requests with no useful prefix match, absorbs bursts directed toward a specialized cache, and remains available during cache migration or catalogue inconsistency. This motivates the explicit constraint $\pi_0\ge \pi_{\min}$ discussed after~\eqref{eq:projection}. The fallback fraction should increase with workload burstiness, telemetry delay, and uncertainty in prefix popularity.

Cache misses and stale metadata do not threaten model correctness if every reusable state is verified by model identifier, tokenizer identifier, adapter version, prefix hash, position encoding, and cache epoch before use. A failed verification simply converts the request to a smaller hit or a full prefill. The same fallback is used when a node failure removes an advertised profile. Thus, spatial prefix caching can degrade gracefully: stale information reduces acceleration but need not change the generated result.

\noindent\textbf{State lifecycle and calibration.}
The profile-construction process in Section~V can be implemented with a radix-tree catalogue. For each vertex $v$, the serving platform records request count, descendant count, reusable depth, state bytes, measured prefill saving, and a time-decayed popularity estimate. A practical value score is
\begin{equation}
\widetilde\Gamma_v(t)=
\frac{\widehat p_v(t)\,\widehat V_v^{\rm pf}}
{B_v+\lambda_{\rm mig}C_v^{\rm mig}+\lambda_{\rm risk}C_v^{\rm var}},
\label{eq:lifecyclevalue}
\end{equation}
where $C_v^{\rm mig}$ is migration cost and $C_v^{\rm var}$ penalizes rapidly varying popularity. The parameters $\lambda_{\rm mig}$ and $\lambda_{\rm risk}$ prevent the controller from repeatedly moving large states for short-lived demand spikes.

The cache lifecycle contains four actions. \emph{Admission} decides whether a newly popular subtree should become persistent. \emph{Replication} increases its spatial density when communication distance or tier load is excessive. \emph{Demotion} shortens or removes a profile when runtime concurrency becomes scarce. \emph{Migration} relocates a state when demand shifts geographically. These actions operate at the intermediate time scale and can be rate limited by a migration budget
\begin{equation}
\sum_{s,t}m_{s\rightarrow t}C_{s\rightarrow t}^{\rm mig}\le C_{\rm epoch}^{\rm mig},
\end{equation}
where $m_{s\rightarrow t}$ is the number or density of nodes changing from profile $s$ to $t$ in one control epoch.

The resulting profile catalogue is useful only when its analytical quantities are tied to the serving stack and workload trace.
The analytical quantities are deliberately modular so that normalized parameters can be replaced by measurements. Table~\ref{tab:calibration} summarizes a reproducible calibration workflow. KV bytes per token should be measured from the actual model architecture, numerical precision, and serving engine rather than inferred only from nominal layer dimensions, because block allocation and metadata add overhead. Prefill function $g(n)$ should be profiled over residual length, batch size, and chunk size. Runtime concurrency should be measured under the same admission policy used in deployment. Finally, the queue model should be validated against trace-driven continuous batching; if Erlang-C is inaccurate, a learned delay map $\widehat W_s(\lambda,\boldsymbol\omega,c_s)$ can replace it without changing the spatial association derivation.

\begin{table}[t]
\caption{Mapping from Analytical Quantities to Measurable Inputs}
\label{tab:calibration}
\centering
\scriptsize
\begin{tabular}{p{1.45cm}p{2.1cm}p{3.35cm}}
\toprule
Quantity & Measurement source & Recommended procedure \\
\midrule
$B_s$ & KV allocator & Record persistent bytes for each ancestor-closed profile, including block and metadata overhead. \\
$g(n)$ & GPU profiler & Sweep residual prompt length, batch size, and chunked-prefill configuration; fit or interpolate. \\
$c_s$ & Admission controller & Measure maximum stable active sequences after reserving profile $s$. \\
$D_q$ & Serving trace & Measure decode occupancy and other execution time by request class. \\
$p_q,H_{s,q}$ & Prompt trace & Tokenize with the deployed tokenizer and construct a time-decayed radix forest. \\
$\xi,\alpha$ & Radio telemetry & Fit transfer delay versus distance, bandwidth allocation, and radio condition. \\
$F_{W_s}$ & Queue telemetry & Compare Erlang-C with discrete-event or production traces and substitute a calibrated tail model if required. \\
\bottomrule
\end{tabular}
\end{table}

To support reproducibility, an empirical paper should report the model and tokenizer, GPU type, numerical precision, serving engine, block size, batching and chunking policy, prompt-trace source, cache lifetime, catalogue-update period, and all fitting errors. Reporting only the cache-hit ratio is insufficient because two deployments with the same hit ratio can have different hit depths, concurrency losses, and queue distributions.

\noindent\textbf{Design implications.}
The analysis leads to several implementation rules. First, replicate shallow common prefixes broadly because they serve many descendants at modest memory cost. Second, deploy deep specialized prefixes selectively and monitor their area service capacity; a deep tier with insufficient density becomes both a radio-distance and queueing bottleneck. Third, reserve fallback capacity instead of allocating all memory to persistent states. Fourth, route by marginal end-to-end latency saving rather than longest-prefix-first. Fifth, include telemetry age and uncertainty in the routing cost. Sixth, update profile placement more slowly than request routing so that the network does not chase transient popularity noise. Finally, plan edge density and profile diversity jointly: additional GPUs cannot compensate for a poor profile mix, and deeper caching cannot compensate for insufficient aggregate service capacity.

\subsection{Extensions Beyond the Baseline Model}\label{sec:extensions}
The baseline deliberately uses an equivalent distance-dependent communication cost, local persistent states, stationary Poisson traffic, and independent profile marks. These assumptions make the central coupling tractable, but the derived architecture only requires candidate-specific communication, reuse, and delay estimates. The following extensions therefore replace individual modules rather than the overall formulation.

\noindent\textbf{Wireless and disaggregated state access.}
The delay-equivalent model in~\eqref{eq:commdelay} isolates the cache--compute tradeoff, but a radio-specific study can replace it by
\begin{equation}
T^{\rm com}_{x,q}=
\frac{U_q}{W^{\rm u}_{x,q}\log_2(1+\mathrm{SINR}^{\rm u}_{x,q})}
+\frac{F_q}{W^{\rm d}_{x,q}\log_2(1+\mathrm{SINR}^{\rm d}_{x,q})},
\label{eq:sinrdelay}
\end{equation}
where $U_q$ is the uploaded prompt payload and $F_q$ is the response payload required before the first token is delivered. Under Rayleigh fading and stationary interferers, conditional success probabilities can be evaluated using the PPP probability-generating functional. The principal complication is that association then depends on random fading and shared bandwidth in addition to distance. Two tractable approaches are possible. A long-term association policy can use mean spectral efficiency conditioned on distance, after which the present derivation applies with a fitted $T^{\rm com}(r)$. Alternatively, an instantaneous policy can condition on the marked channel state and average the selection event over fading, yielding a higher-dimensional but conceptually identical guard-region integral.

Bandwidth sharing introduces another load coupling because the communication delay depends on the number of radio users associated with a node. This can be incorporated by augmenting the fixed point with an average radio load $n_s^{\rm r}$ and replacing $\xi$ by $\xi_s(n_s^{\rm r})$. The resulting controller balances two queues: a radio-transfer queue and a GPU-execution queue. Such a model is especially relevant when prompts or multimodal inputs are large; for short text prompts, GPU queueing and prefill often dominate.

The same marginal-latency principle extends to disaggregated KV retrieval and state transfer.
The baseline assumes that a selected node locally stores the useful prefix state. In a disaggregated architecture, the compute node may fetch KV blocks from a memory pool or another edge node. The realized latency becomes
\begin{equation}
T_{x,q}=T_{x,q}^{\rm com}+T_{x,q}^{\rm fetch}+W_x+T_{x,q}^{\rm pf},
\end{equation}
where
\begin{equation}
T_{x,q}^{\rm fetch}=\frac{B_{x,q}^{\rm KV}}{R_{x}^{\rm bh}}+T_{x,q}^{\rm decmp}+T_x^{\rm cont}
\end{equation}
accounts for state size, backhaul rate, decompression, and interference with foreground inference. Mooncake, MemServe, KVDirect, CacheGen, and ShadowServe demonstrate that these terms can be comparable to the computation saved~\cite{qin2024mooncake,hu2024memserve,chen2024kvdirect,liu2023cachegen,xiang2025shadowserve}. The spatial framework still applies if each candidate is marked by a local or remote-state retrieval mode. A remote prefix is useful only when its prefill saving exceeds the additional radio, backhaul, decompression, and queueing cost.

A richer placement problem may separate compute profiles from memory profiles. Compute nodes and KV-memory nodes can be modeled as two point processes, and a request selects a compute--memory pair. This produces a bipartite spatial matching problem with two distances and two load processes. The one-process model in this paper is a tractable first step and a baseline for evaluating whether disaggregation is worth its networking cost.

\noindent\textbf{Serving dynamics and coordinated placement.}
Production LLM serving uses continuous batching, chunked prefill, prefill/decode disaggregation, and request migration~\cite{agrawal2024sarathi,zhong2024distserve,patel2023splitwise,sun2024llumnix}. Service times are not exponential, arrivals may be bursty, and jobs interact through dynamic batching. Erlang-C is therefore an analytically convenient approximation rather than a claim of exact server behavior.

The spatial results require only a profile-specific advertised delay distribution, not the Markov property itself. One replacement is an $M/G/c$ or phase-type approximation fitted to measured service times. Another is a trace-driven simulator that returns $F_{W_s}(w\mid \lambda_s^{\rm q},\boldsymbol\omega_s,c_s)$. A differentiable surrogate can then be inserted into the fixed point and outer placement solver. For bursty arrivals, a safety factor based on the squared coefficient of variation or a percentile queue predictor can enlarge the robust margin $\psi_x$ in~\eqref{eq:practicalcost}.

Beyond short-term queue dynamics, profile popularity and user geography also evolve across control epochs.
Independent marking represents randomized or decentralized profile assignment. A coordinated controller could deliberately repel identical deep profiles to improve geographic diversity. Poisson-hole, determinantal, or perturbed-lattice models can represent such deployments, although independent void probabilities are lost. The PPP result remains a useful baseline for irregular deployments, node failures, and randomized reconfiguration.

Prompt popularity and user geography vary over time. A dynamic formulation can introduce profile-transition variables and migration cost, producing a multi-period optimization
\begin{align}
\max_{\{\boldsymbol\pi(t)\}}\quad &\sum_t P_{\rm cov}^{\rm TTFT}(\tau_0;t)
-\lambda_{\rm mig}\mathcal C(\boldsymbol\pi(t),\boldsymbol\pi(t-1))
\end{align}
subject to per-epoch stability and migration budgets. Mobility can be handled by predicting the distribution of the user's future serving region during the request lifetime. This is most relevant for long agent sessions and multi-turn conversations, where preserving a warm state at the current node may compete with migrating it toward the user's next region.

\noindent\textbf{Heterogeneity, privacy, and validation.}
An edge network may host several model sizes, quantization levels, adapters, or tenant-specific privacy domains. Model identity can be appended to the cache profile mark. A request then competes only among compatible nodes, while Theorem~1 applies to the reduced candidate set. When a prefix is reusable across adapters only up to a shared base-model layer, the hit function can be decomposed into reusable and adapter-specific components.

Privacy does not require centralizing raw prompt text. Placement can operate on aggregated prefix hashes, depths, popularity counts, state sizes, and lifetimes. However, rare prefix identifiers can still leak information through frequency or membership inference. A practical system should aggregate statistics over privacy domains, apply minimum-support thresholds, and avoid globally advertising low-frequency prefixes. Encryption protects transferred state, but it does not eliminate side channels from catalogue membership or access patterns.

These extensions ultimately require a hardware- and trace-calibrated validation path.
The numerical results validate the analytical machinery and reveal structural mechanisms; they do not claim hardware-specific latency gains. A publication-grade empirical extension should complete four levels of validation. First, verify the prefix forest and overlap-depth distribution on public or application-specific traces. Second, measure $g(n)$, $B_s$, $c_s$, and $F_{W_s}$ using a production serving engine. Third, replay the workload through a multi-node discrete-event simulator with realistic radio and backhaul traces. Fourth, validate selected operating points on a physical testbed and report TTFT percentiles, throughput, memory fragmentation, cache migration traffic, and catalogue staleness.

The principal modeling risks are also clear. Exact prefix reuse may be sparse in some workloads; state lifetime may be shorter than migration time; radio transfer may dominate for multimodal requests; and cache affinity may be correlated across neighboring nodes. These cases reduce the quantitative gain but do not invalidate the analytical question. The framework is most valuable when reusable prompt regions are long, prefill is a material fraction of TTFT, and multiple edge nodes offer meaningfully different cache states and loads.

\endgroup
\section{Conclusion}
This paper introduced a stochastic-geometry and queueing framework for spatial prefix caching in wireless edge LLM networks. A prefix forest and ancestor-closed cache profiles capture nested prompt reuse, while independent spatial marking yields tractable profile-specific PPPs. The derived association and TTFT-coverage expressions jointly account for communication distance, partial prefill reuse, queueing, and GPU-memory-limited concurrency.

The framework changes the central caching question. Classical wireless caching asks whether requested content is available nearby. Spatial prefix caching asks how much computation can be reused nearby, whether that reuse is worth its communication and queueing cost, and how much runtime capacity must be sacrificed to preserve it. The analysis shows that the nearest edge GPU need not minimize TTFT, that additional persistent caching can hurt latency by reducing concurrency, and that queue-aware routing is essential to prevent cache-affinity hotspots. Numerical validation further demonstrates that the optimal profile mix depends on offered load and that edge density alone cannot replace cache diversity.

The resulting design principle is hierarchical: replicate shallow prefixes broadly, deploy deep prefixes selectively, preserve high-concurrency fallback capacity, and route according to marginal end-to-end latency saving. These results provide a foundation for edge-GPU density planning, profile placement, and communication--computation-aware association in future wireless edge intelligence systems. A hardware-calibrated extension using real prompt traces and production serving engines is the next step toward deployment-level validation.

\appendices
\begingroup
\footnotesize
\setlength{\abovedisplayskip}{4pt plus 1pt minus 1pt}
\setlength{\belowdisplayskip}{4pt plus 1pt minus 1pt}
\setlength{\abovedisplayshortskip}{2pt}
\setlength{\belowdisplayshortskip}{2pt}
\section{Detailed Proof of Theorem 1}
We prove the profile-association probability in a constructive way. Consider a typical request of type $q$ at the origin and fix a candidate profile $s$ with $\lambda_s>0$. Let $R_s$ denote the distance to the nearest profile-$s$ node. By the nearest-neighbor law of a PPP,
\begin{equation}
 f_{R_s}(r)=2\pi\lambda_s r\exp(-\pi\lambda_s r^2),\qquad r\ge 0.
\end{equation}
Condition on the event $R_s=r$. Under this conditioning, the nearest profile-$s$ node lies on the circle of radius $r$, and there is no other profile-$s$ node in the disk $\mathcal B(0,r)$. The conditioned profile-$s$ candidate is selected if and only if its predicted cost does not exceed the cost of the nearest node in every competing profile $t\neq s$, i.e.,
\begin{equation}
\xi r^\alpha+\Delta_{s,q}\le \xi R_t^\alpha+\Delta_{t,q},\qquad \forall t\neq s.
\label{eq:app_assoc_start}
\end{equation}
Rearranging~\eqref{eq:app_assoc_start} yields
\begin{equation}
R_t^\alpha\ge r^\alpha+\frac{\Delta_{s,q}-\Delta_{t,q}}{\xi}.
\end{equation}
Since distances are nonnegative, the corresponding guard radius for profile $t$ is
\begin{equation}
G_{s,t,q}(r)=\left[r^\alpha+\frac{\Delta_{s,q}-\Delta_{t,q}}{\xi}\right]_+^{1/\alpha}.
\end{equation}
Therefore, conditioned on $R_s=r$, profile $s$ wins if every competing tier $t$ has no point inside the disk $\mathcal B(0,G_{s,t,q}(r))$. By independent marking, the profile-specific processes are mutually independent PPPs, so the conditional selection probability factors as
\begin{align}
&\Prb\big[s_q^*=s\mid R_s=r\big] \\
&\quad=\prod_{t\neq s}\Prb\big[\Phi_t(\mathcal B(0,G_{s,t,q}(r)))=0\big].
\end{align}
Using the void probability of a PPP,
\begin{equation}
\Prb\big[\Phi_t(\mathcal B(0,u))=0\big]=\exp(-\pi\lambda_tu^2),
\end{equation}
we obtain
\begin{equation}
\Prb\big[s_q^*=s\mid R_s=r\big]=\prod_{t\neq s}\exp\big(-\pi\lambda_tG_{s,t,q}^2(r)\big).
\end{equation}
Multiplying this conditional probability by the density of $R_s$ and integrating over all feasible $r$ gives
\begin{align}
A_{s,q}
&=\int_0^\infty f_{R_s}(r)\Prb\big[s_q^*=s\mid R_s=r\big]dr \\
&=\int_0^\infty 2\pi\lambda_sr\exp\!\left[-\pi\lambda_sr^2-\pi\sum_{t\neq s}\lambda_tG_{s,t,q}^2(r)\right]dr,
\end{align}
which is exactly~\eqref{eq:assocprob}. Since the events $\{s_q^*=s\}$ are mutually exclusive and exhaustive over the active profiles, summing $A_{s,q}$ over $s$ gives one.\hfill$\square$

\section{Detailed Proof of TTFT Coverage}
We derive the TTFT-coverage expression by conditioning on request type, selected profile, and serving distance. Consider request type $q$ and profile $s$. From~\eqref{eq:fsel}, the joint density that the selected node belongs to profile $s$ and lies at distance $r$ is $f_{s,q}^{\rm sel}(r)$. Conditional on this event, the realized TTFT is
\begin{equation}
T_{s,q}=\xi r^\alpha+W_s+T_{s,q}^{\rm pf}.
\end{equation}
For a latency target $\tau$, the success event is
\begin{equation}
T_{s,q}\le \tau
\iff
W_s\le \tau-\xi r^\alpha-T_{s,q}^{\rm pf}.
\label{eq:app_success}
\end{equation}
The right-hand side of~\eqref{eq:app_success} must be nonnegative. This leads to the distance truncation
\begin{equation}
r\le r_{s,q}^{\max}(\tau)=\left[\frac{\tau-T_{s,q}^{\rm pf}}{\xi}\right]_+^{1/\alpha}.
\end{equation}
Whenever $r\le r_{s,q}^{\max}(\tau)$, the conditional success probability equals the waiting-time CDF,
\begin{equation}
\Prb(T_{s,q}\le \tau\mid s_q^*=s,R_s=r)=F_{W_s}(\tau-\xi r^\alpha-T_{s,q}^{\rm pf}).
\end{equation}
Therefore, averaging over distance and then over profile and request type gives
\begin{align}
P_{\rm cov}^{\rm TTFT}(\tau)
&=\sum_q p_q \sum_s \int_0^{r_{s,q}^{\max}(\tau)} f_{s,q}^{\rm sel}(r)
F_{W_s}(\tau-\xi r^\alpha-T_{s,q}^{\rm pf})dr,
\end{align}
which is~\eqref{eq:ttftcoverage}.\hfill$\square$

\section{Erlang-C Waiting-Time Distribution}
Consider a stable $M/M/c$ queue with arrival rate $\lambda$, exponential service rate $\mu$ per server, and offered load $a=\lambda/\mu$. The birth--death chain has stationary probabilities
\begin{equation}
p_n=\begin{cases}
p_0a^n/n!,&0\le n<c,\\
p_0a^n/(c!c^{n-c}),&n\ge c,
\end{cases}
\end{equation}
where normalization gives
\begin{equation}
p_0^{-1}=\sum_{n=0}^{c-1}\frac{a^n}{n!}+\frac{a^c}{c!(1-\rho)},\qquad \rho=\frac{\lambda}{c\mu}<1.
\end{equation}
An arrival waits precisely when it observes at least $c$ jobs. By PASTA, the waiting probability is the stationary tail probability,
\begin{equation}
P^{\rm wait}=\sum_{n=c}^{\infty}p_n
=\frac{\frac{a^c}{c!(1-\rho)}}{\sum_{n=0}^{c-1}\frac{a^n}{n!}+\frac{a^c}{c!(1-\rho)}},
\end{equation}
which is~\eqref{eq:erlangc}. Conditional on waiting, all servers are busy. Owing to exponential service and memorylessness, the workload ahead of the tagged job decreases with net rate $c\mu-\lambda$, so the conditional waiting time is exponential with this rate. Therefore,
\begin{equation}
\Prb(W>w)=P^{\rm wait}\exp[-(c\mu-\lambda)w],\qquad w\ge0.
\end{equation}
The distribution has mass $1-P^{\rm wait}$ at zero, and integration of the tail gives $\E[W]=P^{\rm wait}/(c\mu-\lambda)$, proving~\eqref{eq:meanwait} and~\eqref{eq:waitcdf}.

\section{Reproducibility Note}
The numerical evaluation follows the fixed-point, quadrature, and Monte Carlo procedures specified in Sections~IV--VI. The accompanying source package contains the implementation used to regenerate every analytical and simulation figure. Hardware-specific coefficients can be substituted without changing the spatial derivations.

\endgroup
\balance
\bibliographystyle{IEEEtran}
\bibliography{references}
\end{document}